\documentclass[twocolumn,secnumarabic,amssymb, nobibnotes, aps, prl, superscriptaddress]{revtex4-2}

\usepackage{amsmath,mathrsfs,amssymb,amsfonts,amsthm,mathtools}
\usepackage{graphicx,color}
\usepackage{physics}
\usepackage[T1]{fontenc}
\usepackage{bm}
\usepackage{empheq}
\usepackage{hyperref}
\hypersetup{
    colorlinks=true,
    linkcolor=blue,
    citecolor=blue,
    urlcolor=blue
}

\usepackage{amsmath,mathrsfs,amssymb,amsfonts,amsthm,mathtools}
\usepackage{physics,color}
\usepackage{subcaption}
\usepackage{multirow}
\usepackage{tabularx}
\usepackage{booktabs}
\usepackage[table]{xcolor}
\usepackage{tikz,tikz-cd}
\usetikzlibrary{shapes.geometric, arrows, positioning}
\usepackage{empheq}
\usepackage{bm}

\newtheorem{lemma}{Lemma}

\pdfoutput=1 
\usepackage[T1]{fontenc} 

\newcommand{\nn}{\nonumber}

\def\ft#1#2{{\textstyle{\frac{\scriptstyle #1}{\scriptstyle #2} } }}
\def\fft#1#2{{\frac{#1}{#2}}}

\begin{document}
\title{The Kerr-Induced Superradiant Tricritical Point}

\author{Arash Azizi}
\email[]{sazizi@tamu.edu}
\author{Reed Nessler}
\affiliation{Texas A\&M University, College Station, TX 77843, U.S.A.}

\begin{abstract}
The interplay of interaction and dissipation in open quantum systems can forge phase transitions beyond conventional paradigms. In the canonical quantum Rabi model, we demonstrate that a photon-photon (Kerr) interaction transforms the celebrated superradiant phase transition, inducing a line that separates continuous, second-order behavior from discontinuous, first-order regimes. Strikingly, this is not a line of tricritical points. Instead, we prove that genuine tricriticality is an isolated phenomenon, emerging only at a single, mathematically exact dissipation rate: the ``sweet spot'' of $\kappa_t^2 = (\sqrt{13}-2)/3$. Our discovery pinpoints the precise and restrictive conditions for realizing a rare multicritical point in a foundational quantum optical system, identifying this critical dissipation as a unique gateway to a higher order of criticality.
\end{abstract}

\maketitle
\textit{Introduction.}---%
Quantum phase transitions (QPTs), driven by quantum fluctuations at zero temperature, are a cornerstone of modern many-body physics \cite{sachdev2011quantum, vojta2003quantum}. They mark qualitative changes between distinct ground states and have been observed across diverse platforms, from ultracold atomic gases \cite{greiner2002quantum} to correlated condensed matter systems.

A central arena for exploring collective quantum phenomena is light–matter interaction. The quantum Rabi model (QRM) \cite{Rabi1937} and its $N$-atom extension, the Dicke model \cite{dicke1954coherence}, predict a hallmark superradiant QPT \cite{emary2003quantum, Hwang2015Plenio}, which has been linked to exotic effects including topological phases \cite{Mivehvar2017Topological} and emergent quasicrystalline order \cite{Mivehvar2019Quasicrystalline}. Rapid advances in experimental platforms such as circuit QED \cite{Wallraff2004, Fink2009TC_cQED, forn2019ultrastrong, kockum2019ultrastrong} and trapped ions \cite{Leibfried2003, Zheng2025Experimental} have enabled unprecedented control, motivating studies of anisotropy \cite{Ye2025QRM}, lattice generalizations \cite{Xu2024next_nearest, Xu024Dicke_Dimer}, and nonequilibrium multistability \cite{Ge2024Nori, Mivehvar2024Dicke}.

Realistic quantum optical systems inevitably feature both dissipation and intrinsic nonlinearities. A key example is the Kerr nonlinearity, arising from photon–photon interactions, which is known to drive phenomena from optical bistability to quantum fluids of light \cite{drummond1980quantum, carusotto2013quantum}. The interplay of driving, dissipation, and nonlinearity produces a rich nonequilibrium landscape, including dissipative phase transitions (DPTs) \cite{fitzpatrick2017observation, rodriguez2017probing} and engineered criticality using squeezed baths \cite{Zhu2020Squeezed_Agarwal, zhang2025Timecrystal_Agarwal}. This raises a sharp question: how does a fundamental nonlinearity like Kerr modify the canonical superradiant QPT? While first-order transitions have been seen in related models \cite{Pei2024Squeezed_QRM}, the emergence of \emph{genuine} multicriticality—where continuous and discontinuous lines meet—remains largely unexplored \cite{Zhu2024Nonreciprocal_Nori, Lu2024tricriticality, Xu2019Tricritical}.

In this Letter, we address this question for the dissipative QRM with an intrinsic Kerr interaction. A mean-field analysis reveals a rich phase diagram where the transition can be either continuous (second-order) or discontinuous (first-order), with the latter being characterized by significant bistability and hysteresis. We show that the NP instability threshold is universal—independent of Kerr strength—yet the transition order is strongly Kerr-dependent.  Most strikingly, we identify a single, isolated tricritical point for attractive Kerr interaction at a specific dissipation rate, $\kappa_t \approx 0.732\,\omega_c$. Within a Landau-theory description of the transition—where the equation of state for the order parameter $n\!\equiv\!|\alpha|^2$ reads $\delta(n)=A\,n+B\,n^2+C\,n^3+\cdots$ with $\delta$ the reduced pump distance to threshold—the linear and quadratic coefficients vanish $(A=B=0)$ while $C>0$. Consequently, the scaling changes from $|\alpha|\!\propto\!\delta^{1/2}$ to the tricritical $|\alpha|\!\propto\!\delta^{1/6}$. Our results reveal a concrete mechanism to engineer the order of a fundamental light–matter transition and demonstrate that Kerr physics can induce genuine multicriticality in a canonical quantum-optical system.

\textit{Model and Universal Instability Threshold.}---%
We consider a quantum Rabi model with an intrinsic Kerr nonlinearity, described by the Hamiltonian $H/\hbar = \omega_c a^{\dagger} a + \omega_a S_z + \frac{\lambda}{\sqrt{N}}(S_{+} + S_{-})(a + a^{\dagger}) + K (a^{\dagger})^2 a^2$. The system is coupled to a zero-temperature bath, inducing cavity photon loss at a rate $\kappa$, governed by the Lindblad master equation $\dot{\rho}=-i[H, \rho]/\hbar+\mathcal{L}_{\text {cav}} \rho$, where $\mathcal{L}_{\text{cav}}\rho = \kappa(2a\rho a^\dagger - a^\dagger a\rho - \rho a^\dagger a)$.

In the thermodynamic limit ($N \gg 1$), the mean-field dynamics for the normalized observables $\alpha = \langle a \rangle / \sqrt{N}$ and $(X, Y, Z) = (\langle S_x \rangle, \langle S_y \rangle, \langle S_z \rangle) / N$ govern the system. The equations admit a trivial normal phase (NP) solution ($\alpha=0, Z=-1/2$), which becomes unstable when the coupling $\lambda$ exceeds a critical value $\lambda_c$:
\begin{equation}
\frac{4\lambda_c^2}{\omega_a\omega_c} = 1 + \left(\frac{\kappa}{\omega_c}\right)^2\,.
\label{eq:threshold}
\end{equation}
Remarkably, this instability threshold is \textbf{universal}, being completely independent of the Kerr nonlinearity $U \equiv 2KN/\omega_c$. This universality arises because the Kerr term, being of higher order in the fields, does not contribute to the linearized dynamics around the vacuum state. As shown in Fig.~\ref{fig:quantum_fluctuations}, the quantum fluctuations $\langle \delta\alpha^\dagger \delta\alpha \rangle$ in the normal phase diverge precisely at this mean-field boundary, signaling the onset of the phase transition. While the onset of superradiance is Kerr-independent, its character is not.

\begin{figure*}[t]
    \centering
    \includegraphics[width=\textwidth]{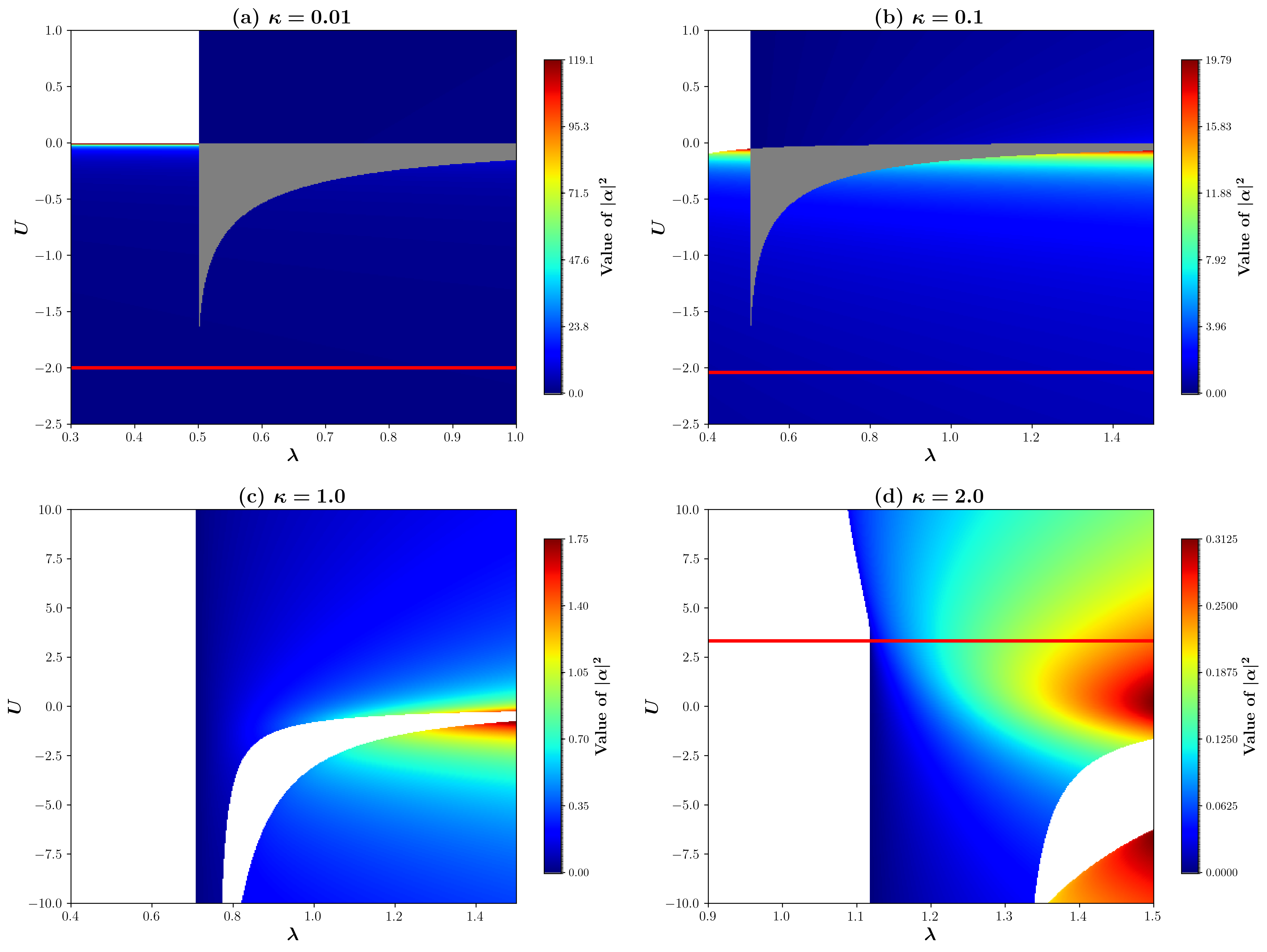}
    \caption{The squared magnitude of the order parameter $|\alpha|^2$ is plotted versus the light-matter coupling $\lambda$ and the Kerr nonlinearity $U$ for four different dissipation rates $\kappa$. The color scale indicates the magnitude. The red line marks $U=U_c(\kappa)$, separating continuous and discontinuous transitions.}
    \label{fig:alpha_abs_sq}
\end{figure*}

\textit{Quantum Fluctuations at the Critical Point.}---%
Beyond the semi-classical stability analysis, we can provide a quantum-level confirmation of the critical point by analyzing the fluctuations around the normal phase. By linearizing the quantum Langevin equations around the NP steady state and solving the associated Lyapunov equation, we derive an exact analytical expression for the steady-state photon number fluctuations $\langle \delta\alpha^\dagger \delta\alpha \rangle$ (see SM for the full derivation). The result is:
\begin{equation}
\langle \delta\alpha^\dagger \delta\alpha \rangle
=\frac{1}{2N}\,
\frac{\lambda^{2}}{\displaystyle \omega_a\omega_c-\frac{4\lambda^{2}}{1+\kappa^{2}/\omega_c^{2}}}\,.
\label{eq:fluctuations}
\end{equation}
Crucially, the denominator of this expression vanishes precisely when the condition for the universal instability threshold, Eq.~\eqref{eq:threshold}, is met. This divergence is the quintessential signature of an approaching second-order phase transition, where critical slowing down leads to an amplification of quantum noise. This behavior is visualized in Fig.~\ref{fig:quantum_fluctuations}, which shows the sharp divergence of fluctuations along the entire critical boundary in the $(\lambda, \kappa)$ plane for the case without Kerr nonlinearity ($U=0$). This analysis thus confirms the location of the universal threshold and reinforces its interpretation as a continuous phase transition driven by quantum fluctuations.

\begin{figure}[t]
    \centering
    \includegraphics[width=\columnwidth]{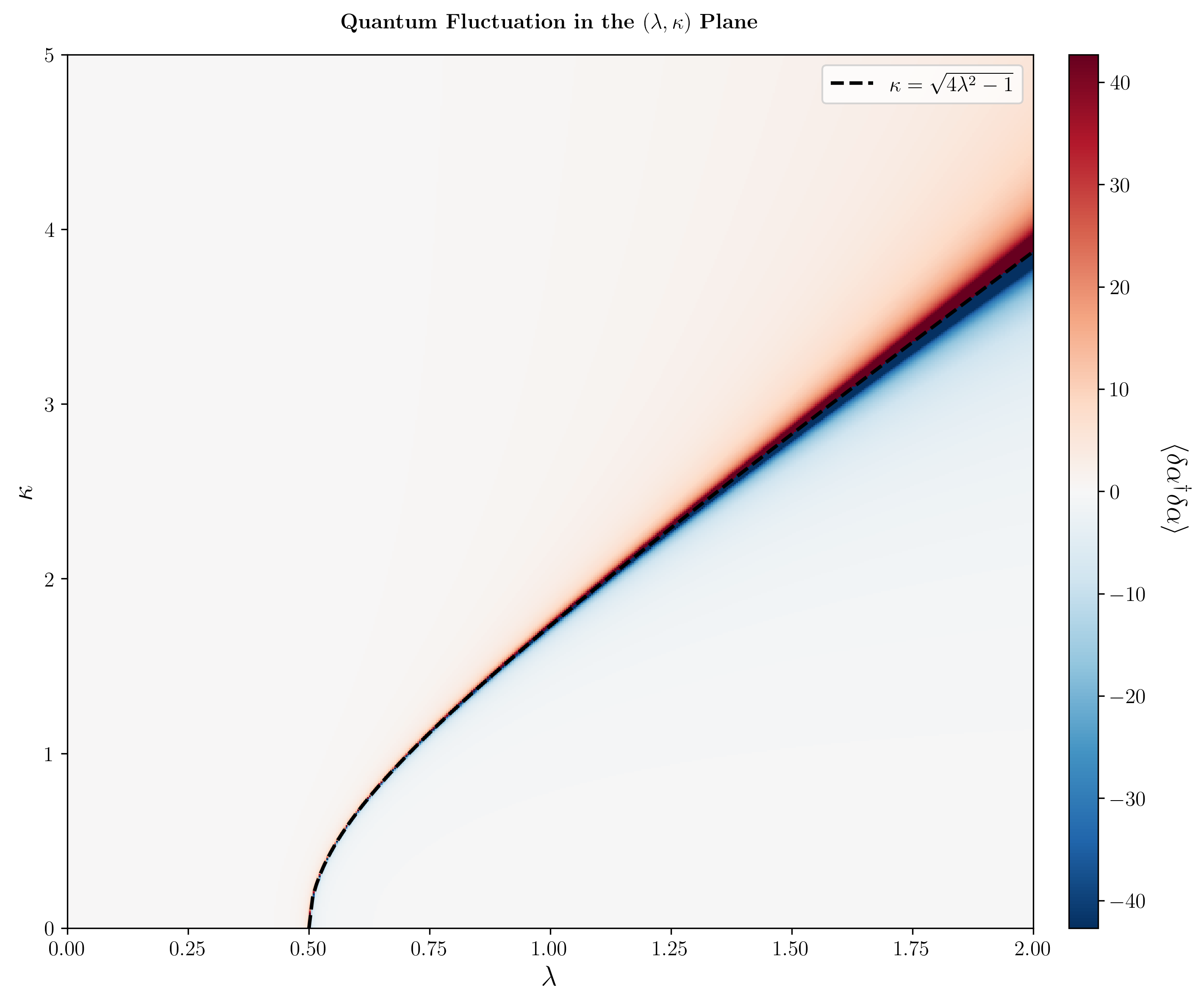}
    \caption{Quantum fluctuation $\langle \delta\alpha^\dagger \delta\alpha \rangle$ in the $(\lambda, \kappa)$ plane for $U=0$. The dashed line marks the mean-field critical point. The divergence of fluctuations signals the onset of the superradiant phase.}
    \label{fig:quantum_fluctuations}
\end{figure}

\textit{Landau Theory and Phase Diagram.}---%
To classify the nature of the transition, we develop a Landau-like theory based on the system's steady-state solutions. We define the order parameter as the photon number $n \equiv |\alpha|^2$ and a control parameter, $\delta \equiv 4\lambda^2 - \omega_a\omega_c(1+\kappa^2/\omega_c^2)$, measuring the deviation from the universal threshold. Near the critical point ($\delta, n \approx 0$), the equation of state can be expanded in a Taylor series:
\begin{equation}
\delta(n) = A n + B n^2 + C n^3 + \mathcal{O}(n^4)\,.
\end{equation}
The signs of the Landau coefficients determine the transition's order and stability. On resonance ($\omega_a=\omega_c=1$), a direct expansion yields (see SM):
\begin{align}
A &= 2(1+\kappa^2) - U(\kappa^2 - 1)\,, \\
B &= \kappa^2 U^2 + 4U + 2(1+\kappa^2)\,, \\
C &= U(-\kappa^2 U^2 + 2U + 2\kappa^2 + 6)\,.
\end{align}
The sign of the leading coefficient, $A$, dictates the nature of the bifurcation. For $A>0$, a stable superradiant solution emerges continuously (a second-order transition). For $A<0$, the emerging solution is unstable, characteristic of a discontinuous, first-order transition accompanied by bistability. This is rigorously confirmed by a full Routh-Hurwitz stability analysis, which shows that the leading-order stability criterion for an emerging SR solution is precisely $A>0$.

The boundary between these two regimes occurs when $A=0$, which defines a critical Kerr nonlinearity $U_c(\kappa) = 2(\kappa^2+1)/(\kappa^2 - 1)$. This equation describes the order-change line in the phase diagram in the panel a of Fig.~\ref{fig:main_result_combined}, separating the first-order (red) and second-order (blue) regions. This theoretical framework is visually confirmed in Fig.~\ref{fig:alpha_abs_sq}, which plots the order parameter versus $\lambda$ and $U$. The red dashed line, representing $U_c(\kappa)$, clearly separates regions of continuous growth from those with abrupt, discontinuous jumps. As illustrated in Fig.~\ref{fig:transition_behavior}, these first-order transitions are accompanied by characteristic hysteresis loops, where the system's state depends on the direction of the parameter sweep.

\begin{figure}[t]
    \centering
    \includegraphics[width=\columnwidth]{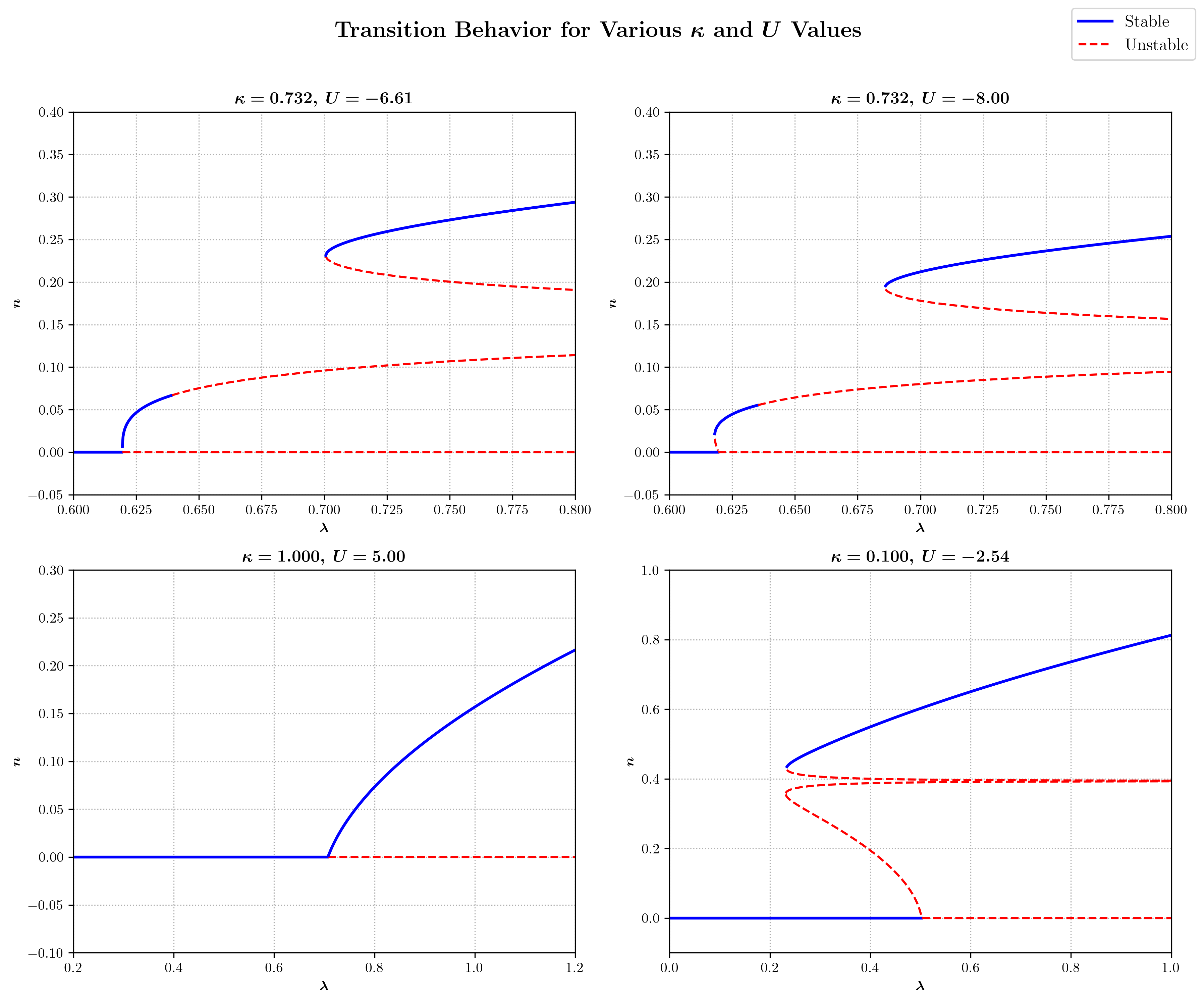}
    \caption{Transition behavior (photon number $n=|\alpha|^2$ vs. $\lambda$) for various $\kappa$ and $U$ values. Solid blue lines indicate stable steady states, while dashed red lines represent unstable ones. The top-left panel shows behavior near the tricritical point. Hysteresis loops, characteristic of first-order transitions, are evident in the top-right and bottom-right panels.}
    \label{fig:transition_behavior}
\end{figure}

\textit{The Isolated Tricritical Point.}---%
A higher-order critical point, a \textbf{tricritical point (TCP)}, emerges if the transition changes its fundamental character, requiring the simultaneous vanishing of the first two Landau coefficients, $A=0$ and $B=0$. Substituting the condition $U=U_c(\kappa)$ into the expression for $B$ yields a unique condition on $\kappa$:
\begin{equation}
3\kappa_t^4 + 4\kappa_t^2 - 3 = 0\,.
\end{equation}
Solving for the physically relevant root gives an exact, isolated value for the tricritical dissipation rate:
\begin{equation}
\kappa_t^2 = \frac{\sqrt{13}-2}{3} \approx 0.535 \implies \kappa_t \approx 0.732\,.
\end{equation}
The corresponding attractive Kerr strength is $U_t = U_c(\kappa_t) \approx -6.61$. At this mathematically precise point $(\kappa_t, U_t)$, marked by a star  in the panel a of Fig.~\ref{fig:main_result_combined}, the system exhibits a genuine TCP. This finding reveals that true tricriticality is not a generic feature of the order-change line but an isolated phenomenon accessible only at a critical dissipation ``sweet spot''.

At the TCP, where $A=B=0$, the equation of state is governed by the next non-zero term, $\delta \approx C n^3$. This implies a distinct critical scaling for the order parameter: $n \propto \delta^{1/3}$. Since the order parameter amplitude is $|\alpha| = \sqrt{n}$, its scaling becomes $|\alpha| \propto \delta^{1/6}$. This unique exponent contrasts sharply with the standard mean-field scaling $|\alpha| \propto \delta^{1/2}$ found along the second-order line. The unique geometry of this higher-order transition is visualized in panel (b) of Fig.~\ref{fig:main_result_combined}, which shows the smooth ramp of the second-order transition and the sharp cliff of the first-order one converging at the TCP.

\begin{figure}[t]
    \centering
    \includegraphics[width=0.7\columnwidth]{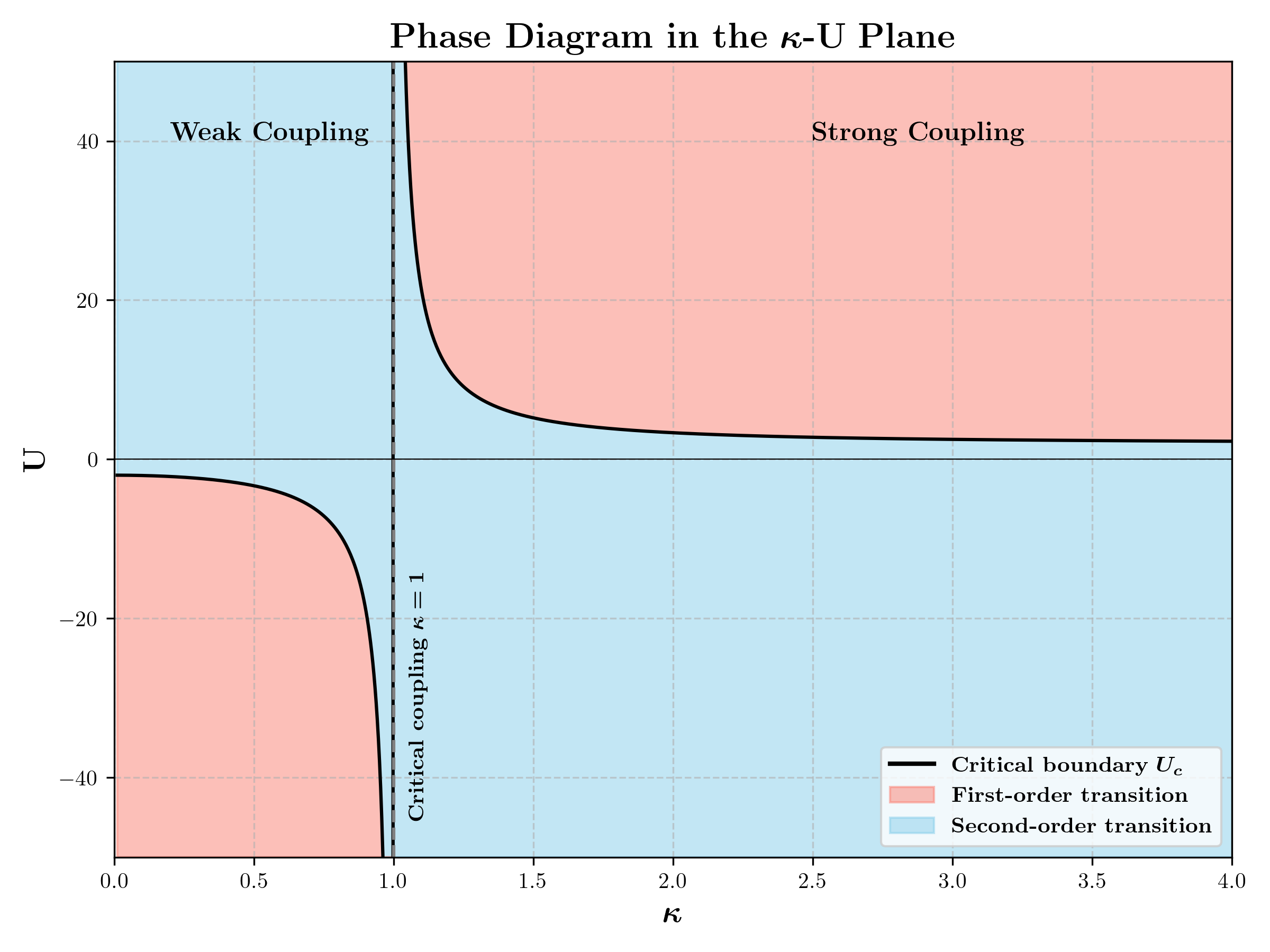}
    \vspace{0.3cm} 
    \includegraphics[width=0.7
    \columnwidth]{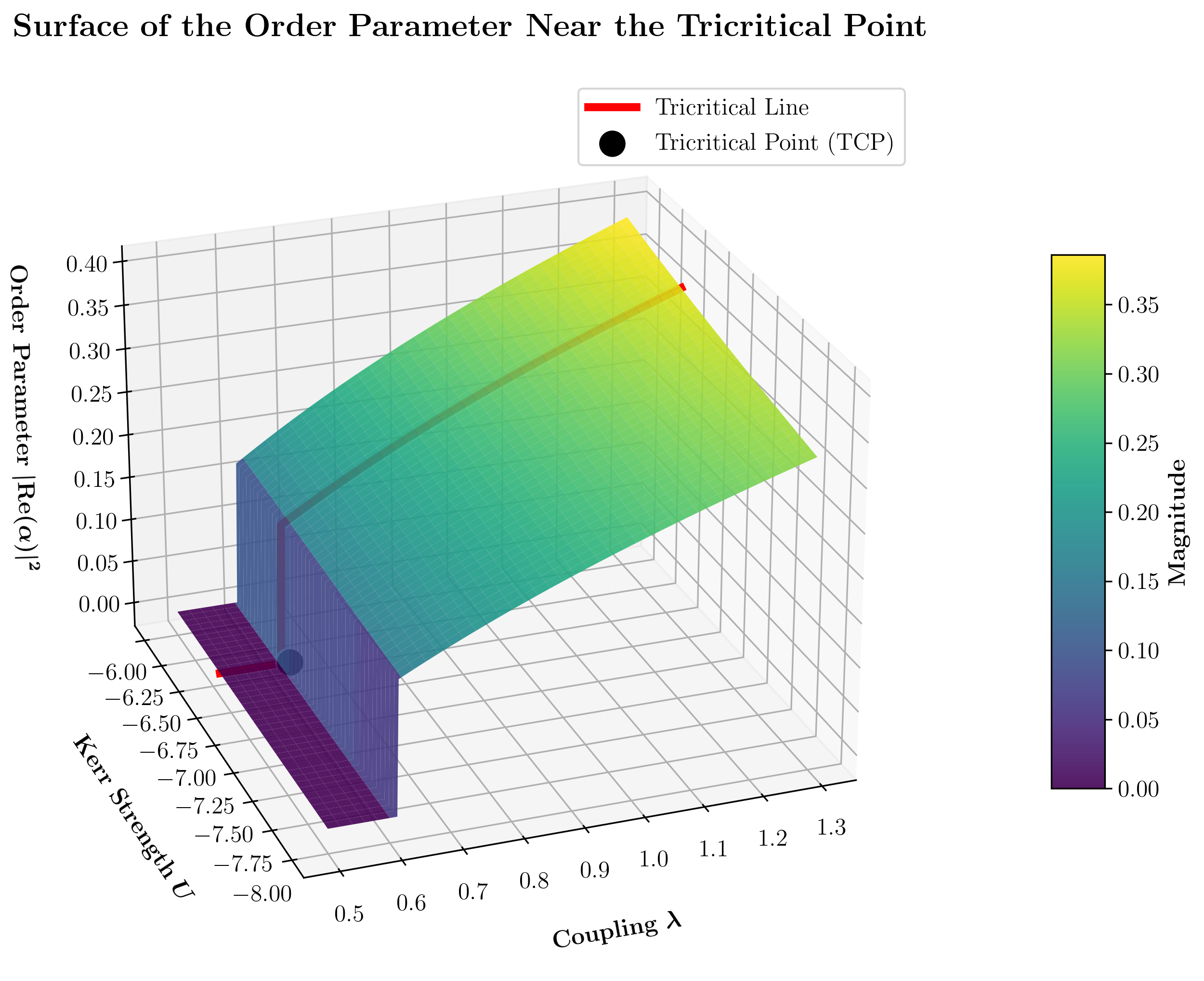}
    \caption{
        (a) Phase diagram in the $\kappa-U$ plane. The black line ($U_c$) separates first-order (red) from second-order (blue) transitions, meeting at the tricritical point (TCP, star).
        (b) Visualization of the order parameter surface at the tricritical dissipation $\kappa_t$. The transition from a smooth second-order ramp (for $U > U_t$) to a sharp first-order cliff (for $U < U_t$) is clearly visible, converging at the TCP.
    }
    \label{fig:main_result_combined} 
\end{figure}

\begin{table*}[t]
\caption{Classification of the Normal-to-Superradiant (NP$\to$SR) phase transition, controlled by the cavity decay rate $\kappa$ and Kerr nonlinearity $U$. The order is determined by the sign of the leading Landau coefficient $A$. The condition $A=0$ marks the boundary where the transition order changes. A true tricritical point requires $A=B=0$. We assume resonance ($\omega_a=\omega_c=1$).}
\centering
\renewcommand{\arraystretch}{1.5} 
\begin{tabular}{@{}lccc@{}}
\toprule
\textbf{Coupling Regime} & \textbf{Interaction Condition ($U$)} & \textbf{Stability ($A$ Sign)} & \textbf{Transition Character} \\
\midrule

\multirow{2}{*}{\begin{tabular}[c]{@{}l@{}}\textbf{Strong Dissipation} \\ ($\kappa > 1$)\end{tabular}}
    & $U < U_c(\kappa)$ (Attractive) & $A>0$ & Second-order (Continuous) \\
    & $U > U_c(\kappa)$ (Repulsive) & $A<0$ & First-order (Discontinuous) \\
\midrule

\multirow{3}{*}{\begin{tabular}[c]{@{}l@{}}\textbf{Weak Dissipation} \\ ($\kappa < 1, \kappa \neq \kappa_t$)\end{tabular}}
    & $U > U_c(\kappa)$ (Attractive) & $A>0$ & Second-order (Continuous) \\
    & $U < U_c(\kappa)$ (Attractive) & $A<0$ & First-order (Discontinuous) \\
    & $U \ge 0$ (Repulsive)         & $A>0$ & Second-order (Continuous) \\
\midrule

\multirow{3}{*}{\begin{tabular}[c]{@{}l@{}}\textbf{Tricritical Dissipation} \\ ($\kappa = \kappa_t$)\end{tabular}}
    & $U > U_t$ & $A>0, B>0$ & Second-order (Continuous) \\
    & $U < U_t$ & $A<0, B>0$ & First-order (Discontinuous) \\
    & $U = U_t$ & $A=0, B=0$ & \textbf{Genuine Tricritical Point} \\
\midrule

\begin{tabular}[c]{@{}l@{}}\textbf{Critical Dissipation} \\ ($\kappa = 1$)\end{tabular}
    & Any $U$ & Always $A>0$ & Second-order (Continuous) \\
\bottomrule

\multicolumn{4}{p{0.9\linewidth}}{\rule{0pt}{2.5ex}\textbf{Note:} The order-change line is defined by $U_c(\kappa) = 2(\kappa^2+1)/(\kappa^2-1)$. A genuine TCP occurs only at the specific dissipation rate $\kappa_t^2 = (\sqrt{13}-2)/3 \approx 0.535$ with the corresponding interaction strength $U_t = U_c(\kappa_t) \approx -6.61$.} \\

\end{tabular}
\label{tab:phase_transition_summary}
\end{table*}

\textit{Conclusion and Experimental Outlook.}---%
In this Letter, we have shown that an intrinsic Kerr nonlinearity can fundamentally restructure the canonical superradiant phase transition, producing an isolated tricritical point at a mathematically exact dissipation rate. Our analysis, summarized in Table~\ref{tab:phase_transition_summary}, reveals that while the normal-phase instability threshold is universal and Kerr-independent, the order of the ensuing normal-to-superradiant transition is fully controlled by the interplay between Kerr nonlinearity and cavity dissipation.

The key physical ingredients for observing this phenomenon—ultrastrong light–matter coupling, tunable Kerr nonlinearity, and engineered dissipation—are all within reach of state-of-the-art experimental platforms. While the specific attractive Kerr strength required for the tricritical point ($U_t \approx -6.61$) is demanding with current technologies, the rapid progress in engineering strong nonlinearities, particularly in superconducting circuits, suggests such parameter regimes may become accessible in the future. Circuit QED systems remain the most promising platform, as they naturally combine large coupling strengths \cite{Wallraff2004,niemczyk2010circuit,yoshihara2017superconducting,forn2019ultrastrong,kockum2019ultrastrong}, in situ Kerr engineering via Josephson elements \cite{rebic2009giant,frattini2017ultrastrong}, and precise control of $\kappa$ through input–output design. The maturity of these platforms for probing criticality is highlighted by recent experiments demonstrating the quantum Rabi phase transition in other systems \cite{Cai2021Single_trapped_ion, Zhang2021Observation} and spontaneous symmetry breaking in superconducting circuits \cite{Zheng2025Experimental}.

An experimental roadmap to observe the predicted tricriticality would begin with mapping the NP$\to$SR phase boundary across several engineered dissipation rates $\kappa$, while monitoring the transition’s nature through hysteresis measurements. At the predicted tricritical ``sweet spot,'' $\kappa_t \approx 0.732\,\omega_c$, which for a typical cavity frequency $\omega_c/2\pi = 5~\mathrm{GHz}$ corresponds to $\kappa_t/2\pi \approx 3.66~\mathrm{GHz}$, the scaling of the photon amplitude is expected to cross over from $|\alpha| \propto \delta^{1/2}$ to the unique tricritical scaling $|\alpha| \propto \delta^{1/6}$. This combination of a universal instability threshold and Kerr–controlled transition order renders the phenomenon robust to microscopic details, making it a compelling and experimentally accessible target for future investigation.

\textit{Acknowledgments.}---%
We are grateful to Girish Agarwal for his foundational contributions in initiating this work and for valuable discussions at the project's outset.  This work was supported by the Robert A. Welch Foundation (Grant No. A-1261) and the National Science Foundation (Grant No. PHY-2013771).

\bibliographystyle{apsrev4-2}
\bibliography{QRM_KerrRef}


\newpage
\clearpage
\onecolumngrid
\appendix
\section*{Supplemental Material:  ``The Kerr-Induced Superradiant Tricritical Point''}

\setcounter{equation}{0}
\renewcommand{\theequation}{S\arabic{equation}}

\section{Master equation}
We consider the quantum Rabi model with Kerr non-linearity. The Hamiltonian is
\begin{align}
& H=H_0+H_1 \nn\\
& H_0=\hbar \omega_c a^{\dagger} a+\hbar \omega_a S_z
+\frac{\hbar \lambda}{\sqrt{N}}\left(S_{+}+S_{-}\right)\left(a+a^{\dagger}\right) \nn\\
& H_1= \hbar K (a^{\dagger})^2 a^2\,,
\end{align}
where $H_0$ is the quantum Rabi Hamiltonian, and $H_1$ represents the Kerr non-linearity.
 The master equation reads
\begin{equation}
\frac{d \rho}{d t}=-\frac{i}{\hbar}[H, \rho]+\mathcal{L}_{\text {cav }} \rho+\mathcal{L}_{\mathrm{A}} \rho\,,
\end{equation}
where $\mathcal{L}_{\text {cav}}$ and $\mathcal{L}_{\text {A}}$ denote cavity and atomic dissipation, respectively. The time derivative of any operator can be found as
\begin{align}
\frac{d}{d t}\langle {\cal O}\rangle=&\frac{d}{d t} \operatorname{tr}({\cal O} \rho) = \operatorname{tr}({\cal O} \frac{d \rho}{d t})=
\operatorname{tr}\left[ {\cal O}\left(-\frac{i}{\hbar}[H, \rho]+\mathcal{L}_{\text {cav}} \rho +\mathcal{L}_{\text {A}} \rho\right) \right]
\nn\\
=&-\frac{i}{\hbar} \,\operatorname{tr}([{\cal O},H] \rho)
+\operatorname{tr}\big( {\cal O}\mathcal{L}_{\text {cav}} \rho \big)
+\operatorname{tr}\big( {\cal O}\mathcal{L}_{\text {A}} \rho \big)\,,
\end{align}
where we have used the cyclicity of the trace, as follows:
\begin{align}
\operatorname{tr}([A, B] C)=\operatorname{tr}([B, C] A)\,.
\end{align}
\subsection{Unitary terms}
We now derive the unitary contribution in the master equation. Using the cyclicity of the trace, $\operatorname{tr}([H, \rho] a)=\operatorname{tr}([a,H] \rho)$. Since
\begin{align}
& {[a, H]=\left[a, H_0\right]+\left[a, H_1\right]\,, \quad\left[a, H_0\right]=\frac{\partial}{\partial a^{\dagger}} H_0=\hbar \omega_c a+\frac{\hbar \lambda}{\sqrt{N}}\left(S_{+}+S_{-}\right)} \nn\\
& [a, H_1]=\hbar K\left[a, (a^{\dagger})^2 a^2 \right]=2 \hbar K a^{\dagger} a^2\,,
\end{align}
then,
\begin{equation}
{[a, H]=\hbar \omega_c a+\frac{\hbar \lambda}{\sqrt{N}}\left(S_{+}+S_{-}\right)}
+ 2 \hbar K a^{\dagger} a^2\,.
\end{equation}
The commutators of the spin operators with the Hamiltonian can be found from the standard commutation relations $\left[S_i, S_j\right]=i \varepsilon_{ijk} \,S_k$ as follows
\begin{align}
\left[S_x, H\right]=&-i \hbar \omega_a S_y\,,\nn\\
\left[S_y, H\right]=&i \hbar \omega_a S_x-\frac{2 i \hbar \lambda}{\sqrt{N}} S_z\left(a+a^{\dagger}\right)\,,\nn\\
\left[S_z, H\right]=&\frac{2 i \hbar \lambda}{\sqrt{N}} S_y\left(a+a^{\dagger}\right)\,.
\end{align}
\subsection{Cavity and atomic dissipation terms}
Here $\mathcal{L}_{\text {cav}}$ and $\mathcal{L}_{\text {A}}$ denote cavity and atomic dissipation, respectively. They are as follows
\begin{align}
\mathcal{L}_{\text {cav}}\, \rho =& \kappa\left(2 a \rho a^{\dagger}-a^{\dagger} a \rho-\rho a^{\dagger} a\right)\,, \nn\\
\mathcal{L}_{\text {A}}\, \rho =&\frac{\gamma}{N}\Big(2 S_{-} \rho S_{+}-S_{+} S_{-} \rho-\rho S_{+} S_{-}\Big)\,.
\end{align}
Note that both $\mathcal{L}_{\text {cav}}\, \rho$ and $\mathcal{L}_{\text {A}}\, \rho$ are defined such that $\operatorname{tr} \left(\mathcal{L}_{\text{cav}}\, \rho\right)=\operatorname{tr} \left(\mathcal{L}_{\text{A}}\, \rho\right)=0$. The following lemma is useful for calculating the dissipation terms.
\begin{lemma}
\begin{align}
 \boxed{\operatorname{tr}\left[\left(\mathcal{L}_A \, \rho\right) {\cal O}\right] =
\frac{\gamma}{N} \operatorname{tr}\Big[\Big(\left[S_{+}, {\cal O}\right] S_{-}
+S_{+} \left[{\cal O},S_{-}\right]\Big) \rho\Big]}\,.
\end{align}
\end{lemma}
\begin{proof}
\begin{align}
\operatorname{tr}\left[\left(\mathcal{L}_A \, \rho\right) {\cal O}\right] =&
\frac{\gamma}{N} \operatorname{tr}\Big[\left(2 S_{-} \rho S_{+}-S_{+} S_{-} \rho-\rho S_{+} S_{-}\right) {\cal O}\Big] \nn\\
=&\frac{\gamma}{N} \operatorname{tr}\Big[\Big(2 S_{+} {\cal O} S_{-}- {\cal O} S_{+} S_{-} -\ S_{+} S_{-}{\cal O}\Big) \rho\Big] \nn\\
=&
\frac{\gamma}{N} \operatorname{tr}\Big[\Big(\left[S_{+}, {\cal O}\right] S_{-}
+S_{+} \left[{\cal O},S_{-}\right]\Big) \rho\Big] \,.
\end{align}
\end{proof}
Henceforth we set $\gamma=0$ (no atomic dissipation) in the mean-field and stability analysis. The same relation holds for $\mathcal{L}_{\text {cav}}$, with $a$ and $a^\dagger$ as annihilation and creation operators, and $\kappa$ as the coupling. Therefore,
\begin{align}
\operatorname{tr}\left[\left(\mathcal{L}_{\text {cav}} \, \rho\right) a\right]=\kappa \operatorname{tr}\Big[\big([a^{\dagger}, a] a
+a^{\dagger}[ a, a]\big) \rho\Big]=
-\kappa \operatorname{tr}(a\rho)=-\kappa\langle a\rangle\,.
\end{align}
For cavity dissipation acting on spin operators,
\begin{equation}
\operatorname{tr}\Big[\left(\mathcal{L}_{\text {cav}} \rho\right) S_i\Big]=\operatorname{tr}\left(\mathcal{L}_{\text {cav}}\left(\rho S_i\right)\right)=0\,,
\end{equation}
since $[S_i, a]=\left[S_i, a^\dagger\right]=0$, and $\operatorname{tr}\left(\mathcal{L}_{\text {cav }}\sigma\right)=0$, for $\sigma=\rho S_i$.
Finally,
\begin{equation}
\boxed{\frac{d}{d t}\langle a\rangle
=-\frac{i}{\hbar} \langle [a, H] \rangle - \kappa \langle a\rangle}
\end{equation}
Since $S_{+}+S_{-}=2 S_x$,
 we have
 \begin{equation}
\frac{d \langle a\rangle}{d t}=-i \omega_c \langle a\rangle
-i \frac{\lambda}{\sqrt{N}} \langle S_{+}+S_{-} \rangle 
-2 i K \,\langle a\rangle^{*}\, \langle a\rangle^2-\kappa \langle a\rangle\,,
\end{equation}
where we have assumed $\langle a^{\dagger} a^2 \rangle=
\langle a\rangle^{*}\, \langle a\rangle^2$. It is convenient to normalize the mean values to eliminate the factor of $N$. Namely,
\begin{align}
\left\langle S_x\right\rangle=N X,\quad
\left\langle S_y\right\rangle=N Y,\quad
\left\langle S_z\right\rangle=N Z,\quad
\langle a\rangle=\sqrt{N} \alpha,\quad
 \alpha=\alpha_{\mathrm{Re}}+i \alpha_{\mathrm{Im}} .
 \end{align}
 Thus, the master equation yields the following mean-field equations for the photon and spin variables:
\begin{align}
 \frac{d \alpha}{d t}&=-i \omega_c \alpha-2 i \lambda X
 -2 i K N \abs{\alpha}^2\, \alpha - \kappa \alpha \\
 \frac{d X}{d t}&=-\omega_a Y, \quad \frac{d Y}{d t}=\omega_a X-2
 \lambda Z\left(\alpha+\alpha^*\right),
\quad \frac{d Z}{d t}=2 \lambda Y\left(\alpha+\alpha^*\right)
\end{align}
Alternatively, one may express the mean-field equations in terms of $Q$ and $P$, where
\begin{equation}
Q=\frac{1}{\sqrt{2}}\left(\alpha+\alpha^*\right)=\sqrt{2} \alpha_{\text{Re}}\,, \quad \quad P=\frac{i}{\sqrt{2}}\left(\alpha^*-\alpha\right)=\sqrt{2} \alpha_{\text{Im}}\,.
\end{equation}
Therefore,
\begin{align}
& \frac{d Q}{d t}=\omega_c P-\kappa Q+K N P\left(Q^2+P^2\right), \\
& \frac{d P}{d t}=-\omega_c Q -\kappa P -2 \sqrt{2} \lambda X-K N Q\left(Q^2+P^2\right), \\
& \frac{d X}{d t}=-\omega_a Y, \\
&\frac{d Y}{d t}=\omega_a X-2 \sqrt{2} \lambda Z Q, \\
&\frac{d Z}{d t}=2 \sqrt{2} \lambda Y Q\,. \label{SS}
\end{align}
\section{Solving the mean-field steady-state equations}
There is a trivial solution for the mean-field equations as follows:
\begin{align}
X=Y=P=Q=0\,, \qquad Z=\pm \ft12\,,
\end{align}
which is called the normal phase. Note that although both positive and negative values of $Z$ are solutions to the steady-state equations, the positive $Z$ is unstable.
In addition to the trivial phase mentioned above, there can be non-trivial solutions to (\ref{SS}), known as the superradiant phase, where the photon number is non-zero, i.e., $P\neq 0$ and $Q \neq 0$.
Assuming $\omega_c\neq 0$ and $\omega_a \neq 0$, the steady-state equations for the mean field (\ref{SS}) can be written as follows:
\begin{align}
P -\tilde{\kappa} Q+ \ft{U}2  P\left(Q^2+P^2\right)&=0\,, \nn\\
Q +\tilde{\kappa} P +2 \sqrt{2} \tilde{\lambda} X+ \ft{U}2  Q\left(Q^2+P^2\right)&=0\,, \nn\\
X-2 \sqrt{2}\, \fft{\tilde{\lambda}}{\beta}\, Z Q&=0\,, \nn\\
Y&=0\,,
\end{align}
where
\begin{align}
\tilde{\kappa}=\fft{\kappa}{\omega_c}\,, \qquad
\tilde{\lambda}=\fft{\lambda}{\omega_c}\,, \qquad
U=\fft{2KN}{\omega_c}\,, \qquad
\beta=\fft{\omega_a}{\omega_c}\,.
\end{align}
Also, $X^2+Y^2+Z^2=\frac{1}{4}$ yields $X^2+Z^2=\frac{1}{4}$.
Introducing the following parameter can be helpful:
\begin{align}
s=U \, \abs{\alpha}^2=\fft{U}2\,(Q^2+P^2)\,. \label{s}
\end{align}
\noindent\textit{From now on, $\kappa$ and $\lambda$ denote $\kappa/\omega_c$ and $\lambda/\omega_c$. Whenever we restore units, we explicitly write $\kappa \to \kappa/\omega_c$ and $\lambda \to \lambda/\omega_c$.}
 Thus, the steady-state equations for the mean field read
\begin{align}
P-\kappa Q+s P =&0\,, \label{P}\\
Q+2 \sqrt{2} \lambda X+s Q+\kappa P=&0\,, \label{Q} \\
X-2 \sqrt{2} \fft{\lambda}{\beta} Z Q=&0\,. \label{X}
\end{align}
Using (\ref{P}) and (\ref{X}) to express $P$ and $X$ in terms of $Q$ and $Z$ yields
\begin{equation}
P=\ft{\kappa}{1+s}\,Q, \quad \quad X=\frac{2 \sqrt{2} \lambda Q}{\beta} Z\,. \label{PQXZ}
\end{equation}
Moreover, substituting the above results into (\ref{Q}), and assuming $Q\neq 0$, yields
\begin{equation}
1+s+\frac{8\lambda^2}{\beta} Z + \frac{\kappa^2}{1+s}=0\,,
\end{equation}
or simply
\begin{equation}
Z=- \frac{\beta}{8 \lambda^2} \frac{(s+1)^2+\kappa^2}{s+1}\,. \label{Z}
\end{equation}

Now using (\ref{s}) and (\ref{PQXZ}) results in
\begin{equation}
Q^2+P^2=Q^2\left(1+\frac{\kappa^2}{(1+s)^2}\right)=\frac{2 s}{U}\,,
\end{equation}
and therefore,
\begin{equation}
Q^2=\frac{2s(s+1)^2}{U\left[(s+1)^2+\kappa^2\right]}\,, \qquad
\qquad
P^2=\frac{2s\kappa^2}{U\left[(s+1)^2+\kappa^2\right]}\,. \label{PQ2}
\end{equation}
Finally, using $X^2+Z^2=\frac{1}{4}$ results in
\begin{equation}
X^2+Z^2=Z^2\left(1+\frac{8 \lambda^2}{\beta^2} Q^2\right)
= Z^2\left(1+\frac{8 \lambda^2}{\beta^2} \,\frac{2s(s+1)^2}{U\left[(s+1)^2+\kappa^2\right]}\right)=\ft14\,,
\end{equation}
where we have used (\ref{PQXZ}) and (\ref{PQ2}). Now, plugging $Z$ from (\ref{Z}) into the above equation yields the following fifth-order polynomial equation for $s$:
\begin{equation}
\boxed{f(s)=\Big((s+1)^2+\kappa^2\Big)\left[\frac{1}{U} s(s+1)^2+\fft{\beta^2}{16 \lambda^2} \,\left((s+1)^2+\kappa^2\right)\right]-\lambda^2(s+1)^2=0}\,. \label{f}
\end{equation}
Expressed in terms of the original quantities appearing in the Hamiltonian, we have
\begin{align}
f(\abs{\alpha}^2)=&\Bigg[\Big(\fft{2KN}{\omega_c} \, \abs{\alpha}^2+1\Big)^2+(\fft{\kappa}{\omega_c})^2\Bigg] \nn\\
& \quad \times \Bigg\{\abs{\alpha}^2\Big(\fft{2KN}{\omega_c} \, \abs{\alpha}^2+1\Big)^2
+\fft{\omega_a^2}{16 \lambda^2} \,\left(\Big(\fft{2KN}{\omega_c} \, \abs{\alpha}^2+1\Big)^2+(\fft{\kappa}{\omega_c})^2\right)\Bigg\} \nn\\
&-(\fft{\lambda}{\omega_c})^2
\Big(\fft{2KN}{\omega_c} \, \abs{\alpha}^2+1\Big)^2=0\,.
\end{align}
\begin{equation}
\alpha_{\text{Re}}^2=\frac{1}{U} \frac{s(s+1)^2}{\left[(s+1)^2+\kappa^2\right]}\,, \qquad
\qquad
\alpha_{\operatorname{Im}}^2=\frac{\kappa^2}{U} \frac{s}{\left[(s+1)^2+\kappa^2\right]}
\end{equation}
One may normalize all parameters with respect to $\omega_c$, and also assume $\omega_c=\omega_a$. Namely,
\begin{align}
\omega_a=\omega_c=1\,, \quad \quad U=\ft{2KN}{\omega_c}=2KN\,,
\end{align}

\subsection{\texorpdfstring{From the SR-side inequality to the existence bound $|Z|\le\tfrac12$}{}}
Recall that for steady states with $Q\neq 0$ we obtained
\begin{equation}
Z(s) \;=\; -\,\frac{\beta}{8\lambda^2}\,
\frac{(1+s)^2+\kappa^2}{\,1+s\,}
\qquad\text{with}\qquad
\beta=\frac{\omega_a}{\omega_c}.
\end{equation}
Let $t:=1+s$. For the cavity Kerr case with $U>0$ we have $t\ge 1$. Then
\begin{equation}
|Z(t)| \;=\; \frac{\beta}{8\lambda^2}\,g(t),
\qquad
g(t):=t+\frac{\kappa^2}{t},\qquad t\ge 1.
\end{equation}

A simple calculus check gives the minimum of $g(t)$ on the domain $t\ge 1$:
\begin{equation}
g_{\min} \;=\;
\min_{t\ge 1}\bigg(t+\frac{\kappa^2}{t}\bigg)
\;=\;
\begin{cases}
1+\kappa^2, & \text{if } \kappa\le 1,\\[2pt]
2\kappa, & \text{if } \kappa\ge 1.
\end{cases}
\end{equation}
Therefore, a \emph{necessary and sufficient} condition for the existence of a physical solution (i.e., for some $t\ge 1$ where $|Z(t)|\le \tfrac12$) is
\begin{equation}
\frac{4\lambda^2}{\omega_a\omega_c} \;\ge\; g_{\min}.
\label{eq:Z-existence}
\end{equation}

Now assume we are on the superradiant (SR) side of the normal-phase threshold, i.e.
\begin{equation}
\frac{4\lambda^2}{\omega_a\omega_c} \;>\; 1+\Big(\frac{\kappa}{\omega_c}\Big)^2.
\label{eq:SR-side}
\end{equation}
Since $1+x^2\ge 2x$ for any $x\ge 0$, the SR-side inequality immediately implies that the existence condition in Eq.~\eqref{eq:Z-existence} is satisfied for both cases of $\kappa$:
\begin{equation}
\frac{4\lambda^2}{\omega_a\omega_c} \;>\; 1+\Big(\frac{\kappa}{\omega_c}\Big)^2 \;\ge\; 2\frac{\kappa}{\omega_c}.
\end{equation}
This confirms that any solution on the SR-side is guaranteed to be physical ($|Z|\le\tfrac12$). A direct consequence is the inequality
\begin{equation}
\boxed{\;\kappa \;<\; \frac{2\lambda^2}{\omega_a}\; }.
\end{equation}

\noindent\emph{Remark.} The normal-phase (NP) \emph{stability} condition is the opposite inequality,
$\;\frac{4\lambda^2}{\omega_a\omega_c}<1+\big(\frac{\kappa}{\omega_c}\big)^2$.
It does \emph{not} by itself imply the existence bound $|Z|\le \tfrac12$; the
implication above uses the SR-side inequality \eqref{eq:SR-side}.

\section{Stability Conditions}
To investigate the stability of the system, we consider small fluctuations around the mean-field quantities and ensure that these fluctuations decay over long times. In other words, for an operator ${\cal O}$, we write $\left\langle {\cal O}\right\rangle\equiv {\cal O}= {\cal O}_0+ \delta {\cal O}$, and $\delta {\cal O}$ should vanish for large $t$.
We start from the mean-field steady-state equations (\ref{SS}) to derive the fluctuations. Although there are five quantities—$X, Y, Z, P$, and $Q$—the angular momentum conservation $X^2+Y^2+Z^2=\ft14$ implies
\begin{equation}
X \delta X+Y \delta Y+Z \delta Z=0\,,
\end{equation}
and since $Y=0$, then $\delta Z=-\frac{X}{Z} \delta X$. Therefore, there are actually four fluctuations to consider. Alternatively, one may derive the stability matrix without using spin conservation. It reads
\begin{align}
\frac{d}{d t} &\left(\begin{array}{l} \label{stab.matrix2}
\delta Q \\
\delta P \\
\delta X \\
\delta Y \\
\delta Z
\end{array}\right) = \\
&\left(\begin{array}{ccccc}
-\kappa+2 K N P Q & \omega_c+K N\left(Q^2+3 P^2\right) & 0 & 0 & 0 \\
-\omega_c - K N\left(P^2+3 Q^2\right) & -\kappa-2 K N P Q & -2 \sqrt{2} \lambda & 0 & 0 \\
0 & 0 & 0 & -\omega_{a} &0 \\
-2 \sqrt{2} \lambda Z& 0 & \omega_{a} & 0 & -2 \sqrt{2} \lambda Q \\
2 \sqrt{2} \lambda Y& 0 & 0 & 2 \sqrt{2} \lambda Q & 0 \\ \nn
\end{array}\right)
\left(\begin{array}{l}
\delta Q \\
\delta P \\
\delta X \\
\delta Y \\
\delta Z
\end{array}\right)\,. \nn
\end{align}
For all fluctuations to vanish at large $t$, all eigenvalues of the matrix $A$ must have negative real parts. To find the eigenvalues, we compute $\operatorname{det}(A-\eta I)$, where $\eta$ denotes the eigenvalues. It reads
\begin{align}
\operatorname{det}(A-\eta I) =&\eta^4+2 \kappa \eta^3 \nn\\
&+\Big[\left(3 KN\left(Q^2+P^2\right)+\omega_c\right)\left(KN\left(Q^2+P^2\right)+\omega_c\right)+\kappa^2+\omega_a^2+8 Q^2{\lambda^2}
\Big]\eta^2\nn\\
&+2 \kappa\left(\omega_a^2+8 Q^2 \lambda^2\right) \eta \nn\\
&+\Big[\left(3 KN\left(Q^2+P^2\right)+\omega_c\right)\left(KN\left(Q^2+P^2\right)+\omega_c\right)+\kappa^2\Big]
\left(\omega_a^2+8 Q^2{\lambda^2}\right) \nn\\
&+8\Big[KN\left(Q^2+3 P^2\right)+\omega_c\Big] \lambda^2 \omega_a Z \,. \label{det}
\end{align}
There are two methods to assess stability. One is to find all eigenvalues $\eta$ and require $\Re(\eta)<0$. However, this approach is not very insightful. Instead, the Routh-Hurwitz stability criterion is more suitable.
\subsection{Routh criterion}
Consider a fourth-order polynomial as follows:
\begin{align}
\eta^4+a \eta^3+b \eta^2+ c \eta+ d \,.
\end{align}
The Routh criterion states that for the real parts of all roots to be negative, the following conditions must hold:
\begin{equation}
a\,,d\,, b-\frac{c}{a}\,, a(b c-a d)-c^2>0\,.
\end{equation}
Let us apply the Routh criterion to the fourth-order polynomial (\ref{det}). Note that $a=2\kappa>0$ implies $c=2 \kappa\left(\omega_a^2+8 Q^2 \lambda^2\right)>0$. Moreover, since $a(b c-a d)-c^2+a^2 d=a c\left(b-\frac{c}{a}\right)$, and given that $a$ and $c$ are positive, requiring $a(b c-a d)-c^2>0$ and $d>0$ ensures $b-\frac{c}{a}>0$. Hence, there are just two additional stability constraints besides $\kappa>0$. Namely,
\begin{align}
d =& \Big[\left(3 K N\left(Q^2+P^2\right)+\omega_c\right)\left(K N\left(Q^2+P^2\right)+\omega_c\right)+\kappa^2\Big]\left(\omega_a^2+8 Q^2 \lambda^2\right) \nn\\
&+8\Big(K N\left(Q^2+3 P^2\right)+\omega_c\Big) \omega_a \lambda^2 Z>0\,, \nn\\
a(b c-a d)-c^2=&-\omega_a \lambda^2 Z\left(K N\left(Q^2+3 P^2\right)+\omega_c\right)>0\,.
\end{align}
It is more illuminating to express the above constraints in terms of $s$ and $Z$:
\begin{align}
\fft{d}{\omega_c^4}=&\frac{\beta^2}{4Z^2}\Big((3 s+1)(s+1)+\kappa^2\Big)
+8 \lambda^2 \beta\,Z\left(s \frac{(s+1)^2+3 \kappa^2}{(s+1)^2+\kappa^2}+1\right) >0\,,
\nn\\
\fft{a(b c-a d)-c^2}{\omega_c^4 \lambda^2}=&-\beta\,Z\left( \fft{s\,\big((s+1)^2+3 \kappa^2\big)}{(s+1)^2+\kappa^2}+1\right)>0\,. \label{stability}
\end{align}
\subsection{Normal phase}
As mentioned earlier, the normal phase is a trivial solution to the steady-state mean-field equation. The stability conditions are straightforward since $s=U \abs{\alpha}^2=0$ and $Z=\pm\ft12$. Thus (\ref{stability}) implies
\begin{align}
\fft{d}{\omega_c^4}=&\beta^2 \big(1+\kappa^2\big)
+8 \lambda^2 \beta\,Z>0\,,
\nn\\
\fft{a(b c-a d)-c^2}{\omega_c^4 \lambda^2}=&-\beta\,Z>0\,, \label{stability-normal}
\end{align}
Restoring $\omega_c$, the second constraint reads $-\fft{\omega_a}{\omega_c}\,Z>0$, so only $Z=-\ft12$ is stable. The first condition, on the other hand, states
\begin{align}
\boxed{\fft{4\lambda^2}{ \omega_a\omega_c}<\Big(1+ \big(\fft{\kappa}{\omega_c}\big)^2 \Big)}
\end{align}
Therefore, while the steady-state equation imposes no constraints on the system parameters $\lambda$, $\omega_a$, $\omega_c$, and $\kappa$, the stability condition does. We have already derived the steady-state equation as in (\ref{f}). Thus, one must solve the fifth-order polynomial (\ref{f}) subject to the two stability constraints (\ref{stability}).
Note that the fifth-order polynomial (\ref{f}) has at most five solutions; however, since $s=U \abs{\alpha}^2$, we are interested in the real solutions. Moreover, for positive (negative) $U$, $s$ is positive (negative), respectively. Finding an analytical expression for the system parameters $\lambda$, $\omega_a$, $\omega_c$, and $\kappa$ that determines whether the system has no, one, or multiple stable solutions is somewhat involved. However, numerical methods can be used to address this.


\section{Landau Theory of the Normal-to-Superradiant Phase Transition}
To analyze the nature of the quantum phase transition, we develop a Landau-type theory grounded in the system's mean-field dynamics. This approach is justified because the transition is characterized by a clear order parameter (the cavity field amplitude $\alpha$), the breaking of a $Z_2$ symmetry, and the presence of a single "soft mode" whose dynamics slow down at the critical point, allowing for a simplified description.

\subsection{Effective Potential and Equation of State}
The most insightful way to derive the Landau expansion is by first constructing an effective potential for the order parameter $n = |\alpha|^2$.

\paragraph{Equilibrium Case ($\kappa = 0$).}
In the absence of dissipation, the system is conservative, and the mean-field steady state corresponds to the minimum of an effective energy function. We set resonance conditions ($\omega_a=\omega_c=1$) for simplicity. By treating the cavity field $\alpha$ as a real classical parameter (by choice of phase) and finding the ground-state energy of the qubit sector in this field, we obtain the mean-field energy per atom:
\begin{align}
E(n)=n+\frac{U}{2}n^2-\frac{1}{2}\sqrt{1+16\lambda^2 n},\qquad n\ge 0.
\end{align}
Here, the terms correspond to the cavity energy, the Kerr nonlinearity, and the minimized spin energy, respectively. The stable steady state is found by minimizing this potential, $\partial_n E(n) = 0$, which yields the equilibrium equation of state (EOS):
\begin{align}
1+U n-\frac{4\lambda^2}{\sqrt{1+16\lambda^2 n}}=0.
\end{align}

\paragraph{Driven-Dissipative Case ($\kappa > 0$).}
In the presence of dissipation, the system reaches a non-equilibrium steady state, and there is no true energy function to minimize. However, the steady-state solution from the full mean-field equations provides an EOS that relates the driving strength $\lambda^2$ to the photon number $n$. From this, we can construct a Lyapunov-like "effective potential" whose minima correspond to the stable steady states.

The central idea is to define a control parameter $\delta$ that measures the deviation from the normal-phase stability threshold, $4\lambda_c^2 = \beta(1+\kappa^2)$. We can express this control parameter as a function of the order parameter $n$:
\begin{equation}
    \delta(n) \equiv 4\lambda^2(n) - \beta(1+\kappa^2)\,.
\end{equation}
For small $n$, this function can be expanded in a Taylor series, which is the Landau normal form for the EOS:
\begin{equation}
    \delta(n) = A(\kappa,\beta,U)\,n + B(\kappa,\beta,U)\,n^2 + C(\kappa,\beta,U)\,n^3 + \mathcal{O}(n^4)\,.
\end{equation}
The effective potential can then be formally defined as $\Phi(n) = \int_0^n \delta(s) ds$, and its minima correctly identify the stable photon numbers.

\subsection{Landau coefficients, cubic term, and stability}

A systematic expansion of the full mean-field steady-state equations allows for the exact calculation of the Landau coefficients. They are as follows
\begin{align}
A(\kappa,\beta,U)&=2\,(1+\kappa^2)-\beta\,U\,(\kappa^2-1)\,, \nn\\
B(\kappa,\beta,U)&=\beta\,\kappa^2\,U^2+4\,U+\frac{2(1+\kappa^2)}{\beta}\,, \nn\\
C(\kappa,\beta,U)&=\frac{U}{\beta}\Big(-\beta^2\kappa^2 U^2+2\beta U+2\kappa^2+6\Big)\,.
\end{align}
The sign of the leading coefficient $A$ determines the nature of the transition: $A>0$ corresponds to a continuous (second-order) onset of the superradiant phase, while $A<0$ signals a discontinuous (first-order) transition. This conclusion is directly consistent with the dynamical stability analysis governed by the Routh–Hurwitz criteria derived earlier.

On resonance ($\beta=1$) the coefficients simplify to
\begin{align}
A&=2(1+\kappa^2)-U(\kappa^2-1)\,, \nn\\
B&=\kappa^2U^2+4U+2(1+\kappa^2)\,, \nn\\
C&=-\kappa^2U^3+2U^2+(2\kappa^2+6)U\,. \label{eq:A,B,C}
\end{align}

\paragraph{Connection to stability analysis.}
The stability of a superradiant solution is determined by the Routh–Hurwitz criterion $S_1>0$. Substituting the steady-state solution for $\lambda^2(n)$ into the full expression for $S_1$ yields 
\begin{align}
\frac{S_1}{n}= & \left[(1+Un)^2+\kappa^2\right]\,
\left[8n(1+Un)^2+4(1+Un)\sqrt{1+4n^2(1+Un)^2}\right] \nn\\
& +2 U(1+Un)\left[8n^2(1+Un)^2+1+4n(1+Un)\sqrt{1+4n^2(1+Un)^2}\right] \nn\\
& -\frac{2 \kappa^2 U}{1+Un}\, .
\end{align}
where we have assumed $\beta=1$. Expanding this expression for small $n$ gives
\begin{equation}
\frac{S_1}{n} = c_0 + c_1 n + c_2 n^2 + \mathcal{O}(n^3)\,,
\end{equation}
where
\begin{align}
    c_0 &= 4(1+\kappa^2) + 2U(1-\kappa^2)\,, \nn\\
    c_1 &= 8(1+\kappa^2) + 20U + 4U\kappa^2 + 2U^2(1+\kappa^2)\,, \nn\\
    c_2 &= 8(1+\kappa^2) + 48U + 16U\kappa^2 + 28U^2 - 2U^3\kappa^2\,.
\end{align}
The leading-order term $c_0$ is directly proportional to $A$, i.e.,
\begin{align}
    c_0 =2A,
\end{align}
linking the Landau expansion to the physical stability of the system. For $A>0$, the emerging solution is stable ($S_1>0$), corresponding to a continuous (second-order) supercritical bifurcation. Conversely, for $A<0$, the solution is unstable ($S_1<0$), corresponding to a discontinuous (first-order) subcritical bifurcation. The line separating these regimes is therefore found by setting $A=0$, namely
\begin{align}
U_c(\kappa)=\frac{2(\kappa^2+1)}{\kappa^2-1}\,. \label{eq:Uc}
\end{align}

\subsection{The Tricritical Point}

A tricritical point (TCP) is a special, higher-order critical point where a line of second-order phase transitions meets a line of first-order phase transitions. At the TCP, the character of the bifurcation itself changes from supercritical to subcritical. This occurs when the first two Landau coefficients vanish simultaneously, i.e.\ $A=0$ and $B=0$.

It is important to distinguish a \textit{tricritical} point from a \textit{triple} point. A triple point is where three distinct thermodynamic phases (e.g., solid, liquid, gas) coexist in equilibrium. A tricritical point, by contrast, does not involve a third phase; it is the point where the nature of a single phase transition changes.

Substituting the condition $U=U_c(\kappa)$ from Eq.~\eqref{eq:Uc} into the expression for $B$ from Eq.~(\ref{eq:A,B,C}), we find its value along the order-change line:
\begin{equation}
B\big|_{A=0} = \frac{2(\kappa^2+1)\,\big(3\kappa^4+4\kappa^2-3\big)}{(\kappa^2-1)^2}
\end{equation}
This leads to a condition on $\kappa$ alone:
\begin{equation}
3\kappa_t^4 + 4\kappa_t^2 - 3 = 0 
\quad\Rightarrow\quad 
\kappa_t^2 = \frac{\sqrt{13}-2}{3} \approx 0.535\,.
\end{equation}
This gives a TCP at $\kappa_t \approx 0.732$ with $U_t = U_c(\kappa_t) \approx -6.61$. At this point, the EOS is dominated by the cubic term, $\delta \propto n^3$, leading to a unique critical scaling for the order parameter, $|\alpha| \propto \delta^{1/6}$. Moreover, since $C(\kappa_t,U_t)>0$ (numerically $C\simeq1.95\times10^2$), the quartic Landau potential remains stable at tricriticality.

\subsection{Hysteresis and Spinodals in the First-Order Regime}
For $A<0$ and $C>0$, the transition is first order with bistability and hysteresis. Stationary branches satisfy the equation of state
\[
\delta_{\mathrm{ext}}=\delta(n), \qquad 
\delta(n)\equiv A n + B n^2 + C n^3, \qquad 
\delta_{\mathrm{ext}}\equiv 4\lambda^2-(1+\kappa^2).
\]
Spinodal points are stationary and marginal, given by $\delta_{\mathrm{ext}}=\delta(n)$ together with
\begin{align}
\delta'(n)=A+2Bn+3Cn^2=0,
\end{align}
which yields (for $B^2>3AC$)
\begin{align}
n_{\pm}=\frac{-B\pm\sqrt{B^2-3AC}}{3C}.
\end{align}
The normal phase loses stability at the up–spinodal $\delta_{\uparrow}=\delta(0)=0$, while the superradiant branch loses stability at the down–spinodal $\delta_{\downarrow}=\delta(n_{-})<0$. Hence the bistable window is $\delta_{\mathrm{ext}}\in[\delta_{\downarrow},\,0]$.

Within this window, coexistence occurs when the two stable minima have equal effective potential. The effective potential is defined as
\begin{align}
\Phi(n;\delta_{\mathrm{ext}}) 
= \int_0^n \delta(s)\,ds - \delta_{\mathrm{ext}}\,n
= \frac{A}{2}n^2 + \frac{B}{3}n^3 + \frac{C}{4}n^4 - \delta_{\mathrm{ext}}\,n.
\end{align}
Stationarity enforces $\delta_{\mathrm{ext}}=\delta(n)$. The coexistence condition then proceeds in steps:

\begin{enumerate}
    \item \textbf{Stationarity (Equation of State):}  
    At equilibrium,
    \[
    \frac{\partial \Phi}{\partial n} = 0 \quad\Rightarrow\quad \delta_{\mathrm{ext}}=\delta(n).
    \]

    \item \textbf{Set the Coexistence Condition:}  
    At coexistence, the effective potentials of the two stable states are equal,
    \[
    \Phi(n_{\mathrm{coex}};\delta_{\mathrm{ext}})=\Phi(0;\delta_{\mathrm{ext}}).
    \]
    Since $\Phi(0;\delta_{\mathrm{ext}})=0$, this reduces to
    \begin{align}
    \frac{A}{2}n_{\mathrm{coex}}^2 + \frac{B}{3}n_{\mathrm{coex}}^3 + \frac{C}{4}n_{\mathrm{coex}}^4 
    - n_{\mathrm{coex}}\big(A n_{\mathrm{coex}} + B n_{\mathrm{coex}}^2 + C n_{\mathrm{coex}}^3\big)=0.
    \end{align}

    \item \textbf{Simplify the Condition:}  
    Dividing by $n_{\mathrm{coex}}^2$ (for $n_{\mathrm{coex}}>0$) yields the  coexistence equation
    \begin{align}
    \frac{A}{2}+\frac{2B}{3}\,n_{\mathrm{coex}}+\frac{3C}{4}\,n_{\mathrm{coex}}^{2}=0.
    \end{align}
\end{enumerate}

The nontrivial solution gives $n_{\mathrm{coex}}>0$, with $\delta_{\mathrm{coex}}=\delta(n_{\mathrm{coex}})$. In quasi-static sweeps, the jump in photon number occurs at the spinodals ($\delta_{\downarrow}$ on the downward sweep, $0$ on the upward sweep), not at $\delta_{\mathrm{coex}}$.


\section{Quantum fluctuations in the normal phase (Lyapunov method)}

\paragraph{Setup and linearization.}
We linearize around the normal-phase (NP) fixed point
\[
(Q_0,P_0,X_0,Y_0)=(0,0,0,0),\qquad Z_0=-\tfrac12.
\]
Define $\delta\bm v=(\delta Q,\delta P,\delta X,\delta Y)^{T}$ and write
\begin{equation}
\frac{d}{dt}\,\delta\bm v(t)=A\,\delta\bm v(t)+\bm\eta(t),
\qquad
A=\begin{pmatrix}A_1&A_2\\ A_3&A_4\end{pmatrix},
\end{equation}
with
\begin{align}
A_1=\begin{pmatrix}-\kappa&\omega_c\\ -\omega_c&-\kappa\end{pmatrix},\quad
A_2=\begin{pmatrix}0&0\\ -2\sqrt2\,\lambda&0\end{pmatrix},\quad
A_3=\begin{pmatrix}0&0\\ \sqrt2\,\lambda&0\end{pmatrix},\quad
A_4=\begin{pmatrix}0&-\omega_a\\ \omega_a&0\end{pmatrix}.
\end{align}

\paragraph{Noise normalization.}
Using $\langle a\rangle=\sqrt{N}\,\alpha$, only the cavity is driven by vacuum noise,
\begin{equation}
D=\frac{\kappa}{N}\begin{pmatrix}I_2&0\\[2pt]0&0\end{pmatrix},
\qquad
\langle \eta_i(t)\eta_j(t')+\eta_j(t')\eta_i(t)\rangle=2D_{ij}\,\delta(t-t').
\end{equation}
At $T=0$ atomic noise is absent in this linearized NP treatment.

\paragraph{Lyapunov equation.}
The steady-state covariance $V=\langle \delta\bm v\,\delta\bm v^{T}\rangle$ satisfies
\begin{equation}
A V + V A^{T}=-D.
\label{eq:Lyap}
\end{equation}
Partition $V$ conformally as
\(
V=\begin{pmatrix}V_1&V_2\\ V_2^{T}&V_4\end{pmatrix}
\)
with $V_1^{T}=V_1$, $V_4^{T}=V_4$. Equation \eqref{eq:Lyap} is equivalent to
\begin{align}
& A_1V_1+V_1A_1^{T}+A_2V_2^{T}+V_2A_2^{T}=-\frac{\kappa}{N}I_2, \label{eq:b1}\\
& A_1V_2+V_2A_4^{T}+A_2V_4+V_1A_3^{T}=0, \label{eq:b2}\\
& A_3V_1+A_4V_2^{T}+V_2^{T}A_1^{T}+V_4A_2^{T}=0, \label{eq:b3}\\
& A_4V_4+V_4A_4^{T}+A_3V_2+V_2^{T}A_3^{T}=0. \label{eq:b4}
\end{align}

\paragraph{Block-by-block solution (intermediate steps).}
Parameterize
\[
V_1=\begin{pmatrix}r&s\\ s&u\end{pmatrix},\qquad
V_2=\begin{pmatrix}a&b\\ c&d\end{pmatrix},\qquad
V_4=\begin{pmatrix}\eta&\chi\\ \chi&\xi\end{pmatrix}.
\]

\noindent\textit{Step 1 (from \eqref{eq:b4}).} 
We begin with the equation for the atomic block, $A_4V_4 + V_4A_4^T + A_3V_2 + V_2^T A_3^T = 0$. Using the anti-symmetry of the atomic evolution matrix, $A_4^T = -A_4$, this becomes $A_4V_4 - V_4A_4 + A_3V_2 + V_2^TA_3^T = 0$. We compute the two terms separately:
\begin{align*}
A_4V_4 - V_4A_4 &= \begin{pmatrix} -2\omega_a\chi & \omega_a(\eta-\xi) \\ \omega_a(\eta-\xi) & 2\omega_a\chi \end{pmatrix} \\
A_3V_2 + V_2^TA_3^T &= \begin{pmatrix} 0 & 0 \\ \sqrt{2}\lambda a & \sqrt{2}\lambda b \end{pmatrix} + \begin{pmatrix} 0 & \sqrt{2}\lambda a \\ 0 & \sqrt{2}\lambda b \end{pmatrix}^T = \begin{pmatrix} 0 & \sqrt{2}\lambda a \\ \sqrt{2}\lambda a & 2\sqrt{2}\lambda b \end{pmatrix}
\end{align*}
Summing these and setting the result to the zero matrix gives the component equations:
\begin{align}
\text{(1,1):}\quad & -2\omega_a\chi = 0 &&\Rightarrow \quad \chi=0 \\
\text{(2,2):}\quad & 2\omega_a\chi + 2\sqrt{2}\lambda b = 0 &&\Rightarrow \quad b=0 \\
\text{(1,2):}\quad & \omega_a(\eta-\xi)+\sqrt{2}\lambda a = 0 &&\Rightarrow \quad \omega_a(\eta-\xi) = -\sqrt{2}\lambda a
\end{align}
This first step significantly simplifies the covariance matrices:
\[
V_4=\begin{pmatrix}\eta&0\\ 0&\xi\end{pmatrix}, \qquad
V_2=\begin{pmatrix}a&0\\ c&d\end{pmatrix}.
\]

\noindent\textit{Step 2 (from \eqref{eq:b2}).}
The matrix equation $A_1V_2 + V_2A_4^T + A_2V_4 + V_1A_3^T = 0$ yields four scalar equations.
\begin{align}
\text{(1,1):}\quad & -\kappa a + \omega_c c = 0 \label{eq:S2_11}\\
\text{(1,2):}\quad & \omega_c d + a\omega_a + r\sqrt{2}\lambda = 0 \label{eq:S2_12}\\
\text{(2,1):}\quad & -\omega_c a - \kappa c - d\omega_a - 2\sqrt{2}\lambda\eta = 0 \label{eq:S2_21}\\
\text{(2,2):}\quad & -\kappa d + c\omega_a + s\sqrt{2}\lambda = 0 \label{eq:S2_22}
\end{align}
From \eqref{eq:S2_11}, we immediately find the relation between $a$ and $c$:
\begin{equation}
c = \frac{\kappa}{\omega_c}a\,. \label{eq:c_of_a}
\end{equation}

\noindent\textit{Step 3 (from \eqref{eq:b1}).}
The matrix equation $A_1V_1+V_1A_1^{T}+A_2V_2^{T}+V_2A_2^{T}=-\frac{\kappa}{N}I_2$ gives three distinct scalar equations.
\begin{align}
\text{(1,1):}\quad & -2\kappa r+2\omega_c s=-\frac{\kappa}{N} \label{eq:S1_11}\\
\text{(2,2):}\quad & -2\kappa u-2\omega_c s-4\sqrt{2}\lambda c=-\frac{\kappa}{N} \label{eq:S1_22}\\
\text{(1,2):}\quad & \omega_c(u-r)-2\kappa s-2\sqrt{2}\lambda a=0 \label{eq:S1_12}
\end{align}
Adding equations \eqref{eq:S1_11} and \eqref{eq:S1_22} gives:
\[
-2\kappa(r+u)-4\sqrt{2}\lambda c = -\frac{2\kappa}{N}\,.
\]
Solving for the sum of the diagonal elements of $V_1$ yields the correct expression:
\begin{equation}
r+u = \frac{1}{N} - \frac{2\sqrt{2}\lambda}{\kappa} c\,. \label{eq:r_plus_u_corrected}
\end{equation}
Using \eqref{eq:c_of_a}, this can be expressed in terms of $a$:
\begin{equation}
r+u = \frac{1}{N} - \frac{2\sqrt{2}\lambda}{\omega_c} a\,.  \label{eq:r_plus_u}
\end{equation}

\noindent\textit{Step 4 (Closing the equations and solving for $a$).}
The previous steps have generated a system of linear equations relating the elements of the covariance matrix. The final task is to solve this system for the key cross-correlation term $a = \langle \delta Q \delta X \rangle$. This involves a systematic but lengthy algebraic elimination of the other variables ($s, u, c, d, \eta, \xi$) from the equations derived in Steps 1-3. After performing this elimination, one arrives at the solution for $a$:
\begin{equation}
a = \frac{1}{2\sqrt2 N}\; \frac{\lambda(\kappa^2+\omega_c^2)}{4\lambda^2\omega_c - \omega_a(\kappa^2+\omega_c^2)}\,. \label{eq:a-sol}
\end{equation}

\paragraph{Cavity fluctuation and final expression.}
The photon-number fluctuation in $\alpha$ is read from $V_1$ as
\begin{equation}
\big\langle\delta\alpha^\dagger\delta\alpha\big\rangle
=\frac12\!\left(\langle\delta Q^{2}\rangle+\langle\delta P^{2}\rangle\right)-\frac{1}{2N}
=\frac12\,(r+u)-\frac{1}{2N}=-\frac{\sqrt{2}\lambda}{\omega_c} a \label{eq:fluc}
\end{equation}
Using \eqref{eq:a-sol} and \eqref{eq:fluc} we obtain
\begin{empheq}[box=\fbox]{equation}
\label{eq:deltaa-final}
\big\langle \delta\alpha^\dagger \delta\alpha \big\rangle
=\frac{1}{2N}\,
\frac{\lambda^{2}}{\displaystyle \omega_a\omega_c-\frac{4\lambda^{2}}{1+\kappa^{2}/\omega_c^{2}}}\,.
\end{empheq}

\paragraph{Threshold and consistency.}
The denominator of \eqref{eq:deltaa-final} vanishes at
\begin{equation}
\boxed{\,4\lambda^{2}=\omega_a\omega_c\!\left(1+\frac{\kappa^{2}}{\omega_c^{2}}\right)\, ,}
\end{equation}
which is exactly the Routh–Hurwitz stability boundary for the NP. The divergence of
$\langle \delta\alpha^\dagger\delta\alpha\rangle$ at this point signals the second-order NP\,$\to$\,SR transition. In the lossless limit $\kappa\to0$,
\(
\langle \delta\alpha^\dagger\delta\alpha\rangle
=\frac{1}{2N}\frac{\lambda^{2}}{\omega_a\omega_c-4\lambda^{2}}
\),
as expected.

\end{document}